\documentclass[letterpaper, 10 pt, conference]{ieeeconf}  

\IEEEoverridecommandlockouts                              

\usepackage{mathrsfs}
\usepackage{xcolor}
\usepackage{yhmath}

\usepackage{amsmath,amssymb,amsfonts}
\usepackage{graphicx}
\usepackage{algorithm,algorithmic}

\usepackage{textcomp}

\usepackage{setspace}
\usepackage{lipsum}

\usepackage{mathtools}
\usepackage{enumerate}

\usepackage{caption}
\usepackage{subcaption}
\usepackage{algorithm}
\usepackage{float}
\newtheorem{definition}{Definition}
\newtheorem{theorem}{Theorem}

\newtheorem{lemma}{Lemma}
\newtheorem{corollary}{Corollary}

\usepackage{hyperref}

\makeatletter
\let\NAT@parse\undefined
\makeatother
\usepackage{cite}

\newtheorem{innerexample}{Example}

\newcommand{\abs}[1]{\left\lvert#1\right\rvert}

\newcommand{\R}{\mathbb{R}}

    {\begin{innerexample}[#1]\pushQED{\qed}}%
    {\popQED\end{innerexample}}
\title{\LARGE \bf
Integrator Forwading Design for Unicycles with Constant and Actuated Velocity in Polar Coordinates
}

\author{Miroslav Krstić$^{1}$, Velimir Todorovski$^{1}$, Kwang Hak Kim$^{1}$, and Alessandro Astolfi$^{2}$ 
\thanks{This work was supported in part by the Office of Naval Research under Grant No. N00014-23-1-2376, in part by the Air Force Office of Scientific Research under Grant No. FA9550-23-1-0535, and in
part by the National Science Foundation under Grant No. ECCS-2151525. The results and opinions in this paper are solely of the authors and do not reflect the position or the policy of the U.S. Government or the National Science
Foundation.}
\thanks{$^{1}$ M. Krstić, K. Kim, and V. Todorovski are with the Department of Mechanical and Aerospace Engineering, UC San Diego, 9500 Gilman Drive, La Jolla, CA, 92093-0411, {\tt\small \{krstic,kwk001,vtodorovski\}@ucsd.edu}}%
\thanks{$^{2}$ A. Astolfi is with  Imperial College London, London,
United Kingdom and the University of Rome Tor Vergata, Rome, Italy,
{\tt\small a.astolfi@imperial.ac.uk}}%
}

\begin{document}

\maketitle
\thispagestyle{empty}
\pagestyle{empty}


\begin{abstract}
In a companion paper, we present a modular framework for unicycle stabilization in polar coordinates that provides smooth steering laws through backstepping. Surprisingly, the same problem also allows application of integrator forwarding. In this work, we leverage this feature and construct new smooth steering laws together with control Lyapunov functions (CLFs), expanding the set of CLFs available for inverse optimal control design. In the case of constant forward velocity (Dubins car), backstepping produces finite-time (deadbeat) parking, and we show that integrator forwarding yields the very same class of solutions. This reveals a fundamental connection between backstepping and forwarding in addressing both the unicycle and, the Dubins car parking problems.
\end{abstract}



\section{Introduction}
Transforming the unicycle model from Cartesian to polar coordinates has proven highly effective for control design. In polar coordinates, the state naturally encodes both distance and relative heading to the target, and the singularity at the origin allows one to bypass the Brockett–Ryan-Coron–Rosier conditions \cite{brockett1983asymptotic,coron1994relation, ryan1994brockett}, which prohibit the design of continuous, as well as discontinuous time-invariant stabilizers for the unicycle kinematics.

This insight was first exploited by Badreddin and Mansour~\cite{badreddin1993fuzzy}, who designed a linear state-feedback controller in polar coordinates. Astolfi~\cite{astolfi1999exponential} later characterized its region of attraction, showing it covers an entire half-plane ($x>0$ or $x<0$ for all headings, except a measure-zero set).
Building on this, Aicardi et al.~\cite{aicardi1995} have developed a passivity-based feedback law with bidirectional forward velocity, achieving global asymptotic stabilization. Their analysis relies on a non-strict Lyapunov function and Barbalat’s lemma, which excludes constructive $\mathcal{KL}$ estimates. Han and Wang~\cite{wang24_force_controlled_safestable} have refined this approach with a Lyapunov function defined only on an arbitrarily large, but compact set containing the origin. Restrepo et al.~\cite{restrepo2020leader} have gone further in exploiting the polar coordinates, attaining global exponential stability via a backstepping design with a strict CLF; however, this comes at a cost: it sacrifices modularity, complicates extensions to barrier CLFs, and restricts the unicycle to unidirectional motion, which can produce less efficient parking trajectories.
Other examples of singular transformations are presented in Astolfi~\cite{astolfi1995exponential,astolfi1996discontinuous}. These allow continuous exponential stabilization in the transformed coordinates, though the transformation excludes certain initial conditions—such as states on the $x$-axis or aligned with the target.
Note that, these transformations ensure stability in the transformed coordinates, while in the Cartesian representation only attractivity is obtained

With constant forward velocity, the unicycle reduces to the Dubins vehicle~\cite{dubins1957curves}, a standard model in guidance and pursuit problems.  
Missile guidance is largely dominated by proportional navigation (PN), which ensures zero line-of-sight error but ignores terminal orientation~\cite{zarchan2012tactical,siouris2004missile}. Linear quadratic (LQ) controllers can enforce orientation and provide optimality guarantees~\cite{palumbo2010modern,ryoo2005optimal,ryoo2006time,shaferman2008linear}, but their reliance on linear or linearized dynamics is often very restrictive.

In \cite{todorovski2025_CLF} and \cite{Krstic2025_Dubins} we apply backstepping in polar coordinates to design steering laws for parking, achieving global asymptotic and half-global finite-time stabilization, respectively. In the present work, we develop control laws for both the unicycle and the Dubins vehicle model via integrator forwarding~\cite{mazenc1996adding, sepulchre1997forwarding, krstic2004feedback}, highlighting their connections with our previous backstepping designs. This approach is enabled either by the modular framework introduced in \cite{todorovski2025_CLF} or by fixing the forward velocity, which exposes a system structure naturally suited to integrator forwarding.  

\section{Modular Feedback Design for the Unicycle}
\label{sec:modular_section}
We consider the unicycle model 
\begin{equation}
\dot{x} =  v \cos\theta, \quad 
\dot{y} = v \sin\theta,
\quad \dot{\theta} = \omega \,, \label{eq:unicycle_cartesian}
\end{equation}
where $(x(t),y(t)) \in \R^2$ is the position of the unicycle in Cartesian coordinates, $\theta(t) \in \mathbb{R}$ is the heading angle, $v$ is the forward velocity, and $\omega$ is the angular velocity input. The unicycle can be represented in polar coordinates with the transformation given in Fig.~\ref{fig:unicycle_cord} as
\begin{subequations}\label{eq:unicycle_polar_closed_loop-Gv-1}
    \begin{align}
    \dot{\rho} &= 
    -v \cos\gamma \,, \label{eq:unicycle_polar_rhodot}\\
    \dot{\delta} &= v  \frac{\sin\gamma}{\rho}\,, \label{eq:unicycle_polar_deltadot}\\
    \dot{\gamma} &= v \frac{\sin\gamma}{\rho} -\omega  \label{eq:unicycle_polar_gammadot} \,.
    \end{align} 
\end{subequations}
The polar coordinates are defined as the distance to the origin \(\rho=\sqrt{x^2+y^2}\), the polar angle \(\delta=\text{{\rm atan2}}(y,x)+\pi\), and the line-of-sight (LoS) angle \(\gamma=\delta-\theta\) (see Fig.~\ref{fig:unicycle_cord}). 
The inverse mapping  \((\rho,\delta,\gamma)\mapsto(x,y,\theta)\) is given by \(x=-\rho\cos\delta\), \(y=-\rho\sin\delta\), and \(\theta=\delta-\gamma\).
The unicycle model \eqref{eq:unicycle_cartesian} in Cartesian coordinates does not admit a continuous/discontinuous stabilizers because of the Brockett–Ryan-Coron–Rosier conditions. In polar coordinates, however, the representation \eqref{eq:unicycle_polar_closed_loop-Gv-1} circumvents this obstruction—at the cost of introducing the singularity at $\rho=0$—and thereby allows smooth static feedback laws for $\rho>0$. Yet, when translated back into Cartesian coordinates, discontinuities reappear, through $\text{atan2}(y,x)$ along the line ${x<0,y=0}$.

Next, we recap the modular framework from our paper \cite{todorovski2025_CLF} for designing feedback laws, that is based on decoupling the state $\rho$ from the angular dynamics $(\delta,\gamma)$.

\label{sec-preliminaries}
\subsection{Forward velocity feedback}
\label{sec:forward_velocity_feedback}
For stabilizing the origin of system \eqref{eq:unicycle_polar_closed_loop-Gv-1}, it is not obvious that any feedback law can significantly outperform or simplify the choice
\begin{equation}
\label{eq-basic-v-control}
\fbox{$v = k_1 \rho \cos\gamma$} \,.
\end{equation}
 Under this velocity feedback, the closed-loop subsystem
\begin{align}
\label{eq:unicycle_polar_closed_loop-Gv-2}
\dot{\rho} = - k_1 \rho \cos^2(\gamma)
\end{align}
is such that $\rho$ converges exponentially to zero whenever $\gamma \neq 0$. On the set $\rho>0$, where the cancellation $\rho/\rho = 1$ holds, the control system \eqref{eq:unicycle_polar_closed_loop-Gv-1} reduce to
\begin{subequations}
\label{eq:unicycle_polar_closed_loop-Gv-3}
\begin{align}
\dot{\delta} &= \frac{k_1}{2} \sin(2\gamma)\,, \label{eq:unicycle_polar_closed_loop-Gv-3n} \\
\dot{\gamma} &= \frac{k_1}{2} \sin(2\gamma) - \omega \label{eq:unicycle_polar_closed_loop-Gv-3a}\,,
\end{align}
\end{subequations}
where the steering input $\omega$ must be designed based only on $(\delta,\gamma)$.
A key observation is that the term $\frac{k_1}{2} \sin(2\gamma)$ in \eqref{eq:unicycle_polar_closed_loop-Gv-3a}, which is destabilizing, has to be canceled. To this end, let
\begin{equation}
\label{eq-omega-general}
\fbox{$\omega = \dfrac{k_1}{2} \sin(2\gamma) + \tilde\omega$} \,,
\end{equation}
where $\tilde\omega$ is then designed for the reduced system
\begin{subequations}
\label{eq:unicycle_polar_closed_loop-Gv-3-pass}
\begin{align}
\dot{\delta} &= \frac{k_1}{2} \sin(2\gamma)\,, \label{eq:unicycle_polar_closed_loop-Gv-3n-pass} \\
\dot{\gamma} &= -\tilde\omega \label{eq:unicycle_polar_closed_loop-Gv-3a-pass}\,.
\end{align}
\end{subequations}

\begin{figure}[t]
\centering
\includegraphics[width=0.6\linewidth]{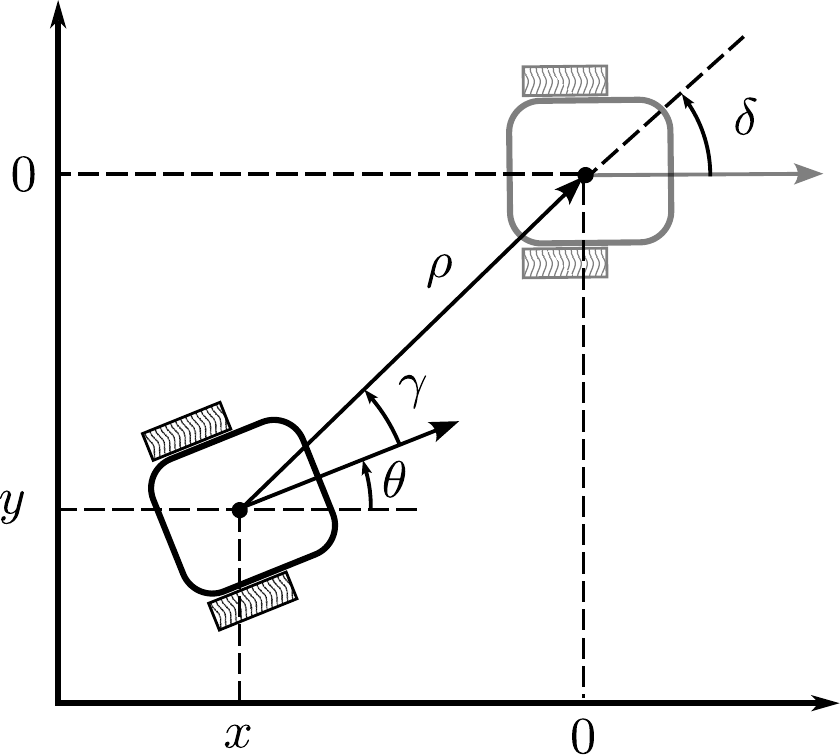}
\caption{Unicycle orientation \((x,y,\theta)\) with respect to the goal state \((0,0,0)\), along with the polar coordinate transformation \((x,y,\theta)\mapsto (\rho,\delta,\gamma)\). 
}
\label{fig:unicycle_cord}
\end{figure}
\subsection{State spaces, stability and CLF definitions}
Before presenting the steering laws $\tilde{\omega}$, we introduce the state spaces 
\begin{equation}
\mathcal{S} := 
\left\{ \rho> 0\right\}\times\mathcal{T}\,,\quad \mathcal{T}:= \left\{ \delta \in\mathbb{R}, \gamma \in\mathbb{R}  \right\} 
\label{eq:ss_S}
\end{equation}
\begin{equation}
    \mathcal{S}_1 := \left\{ \rho> 0\right\}\times\mathcal{T}_1\,,\quad \mathcal{T}_1:= \left\{  \delta \in  \mathbb{R}, \abs{\gamma} <  \pi \right\} \  \label{eq:ss_S1}
\end{equation}
 with the associated metrics
\begin{align}
|(\rho,\delta,\gamma)|_{\mathcal{S}}   &\coloneqq \rho+|(\delta,\gamma)|_{\mathcal{T}} 
      :=  \rho + |\delta| + |\gamma|  
      \\
|(\rho,\delta,\gamma)|_{\mathcal{S}_1} &\coloneqq \rho+|(\delta,\gamma)|_{\mathcal{T}_1} 
      := \rho+ |\delta| + 2\tan\!\frac{|\gamma|}{2}   
\end{align}


\begin{definition}
\label{def-our-GAS}
Consider the system \eqref{eq:unicycle_polar_closed_loop-Gv-2}, \eqref{eq:unicycle_polar_closed_loop-Gv-3-pass} with a feedback law $\tilde\omega(\gamma,\delta)$ that is continuous on a state space $\mathcal{Q}$ with respect to its metric. 
If there exists a class $\mathcal{KL}$ function $\beta$ such that, for all $t\geq 0$, it holds that $|(\rho(t),\delta(t), \gamma(t))|_{\mathcal{Q}}\leq \beta\left(|(\rho_0,\delta_0, \gamma_0)|_{\mathcal{Q}},t\right)$, we say that the 
point $\rho=\delta = \gamma=0$ is {\em globally asymptotically stable on $\mathcal{Q}$} (GAS on $\mathcal{Q}$). 
\end{definition}

With this definition, expressed explicitly in terms of a class $\mathcal{KL}$ bound, we circumvent the problem that stability is defined for equilibria away from the boundary of the state space. 
Next, we define the notion of a CLF for \eqref{eq:unicycle_polar_closed_loop-Gv-1}.
\begin{definition}[CLF for the unicycle \eqref{eq:unicycle_polar_closed_loop-Gv-1}]
\label{def-CLF}
A continuously differentiable function $(\rho,\delta,\gamma)\mapsto V$ 
is a \textit{control Lyapunov function} (CLF) for \eqref{eq:unicycle_polar_closed_loop-Gv-1} if it has the following properties.
\begin{enumerate}
\item There exist class $\mathcal{K}$ functions  $(\bar\alpha_1,\bar\alpha_2)$ such that,  for all $(\rho,\delta,\gamma)$ in 
$\Sigma = \{\rho\geq 0\}\times \hat{\mathcal{T}}$, where $\hat{\mathcal{T}}\in\{\mathcal{T},\mathcal{T}_1\}$, 
$ \bar\alpha_1(|(\rho,\delta, \gamma)|_{\hat{\mathcal{S}}}) \le V(\rho,\delta,\gamma) \le \bar\alpha_2(|(\rho,\delta, \gamma)|_{\hat{\mathcal{S}}})$, where $\hat{\mathcal{S}} = \{\rho> 0\}\times \hat{\mathcal{T}}$. \label{it:clf_prop_1}
\item There exists $(v/\rho,\omega)\in \mathbb{R}^2$ such that
\begin{equation*}
\left[-\dfrac{\partial V}{\partial\rho}\rho\cos\gamma + \left(\dfrac{\partial V}{\partial\delta}+\dfrac{\partial V}{\partial\gamma}\right)\sin\gamma \right]\dfrac{v}{\rho} - \dfrac{\partial V}{\partial\gamma}\omega < 0 \,,
\end{equation*}
for all $(\rho,\delta,\gamma)\neq (0,0,0)$ in $\Sigma$. \label{it:clf_prop_2}
\end{enumerate}
\end{definition}
\raggedbottom
\section{Integrator Forwarding Controllers}

Consider the system \eqref{eq:unicycle_polar_closed_loop-Gv-3-pass}. Observe that \eqref{eq:unicycle_polar_closed_loop-Gv-3n-pass} depends solely on the LoS angle $\gamma$, and since \eqref{eq:unicycle_polar_closed_loop-Gv-3a-pass} involves only the single input $\tilde{\omega}$, the overall system \eqref{eq:unicycle_polar_closed_loop-Gv-3-pass} exhibits a \textit{strict feedforward} structure, enabling the use of the integrator forwarding method~\cite{sepulchre1997forwarding}. Based on this technique, we design two steering controllers accompanied by corresponding CLFs on the state-spaces \eqref{eq:ss_S} and \eqref{eq:ss_S1}, respectively. For this, we define the function $\mbox{sinc}(a)$ as
\begin{equation}
\mbox{sinc}(a) \coloneqq \frac{\sin a}{a} \quad \text{if} \quad a \ne 0, \quad \text{and} \quad \mbox{sinc}(0) = 1 \,,
\end{equation}
which is bounded and continuous, and the sine integral function defined as
\begin{equation}
\mbox{Si}(a) = \int_0^{a}\mbox{sinc}(\alpha){\rm d}\alpha
\,. \label{eq:Si_function}
\end{equation}

\subsection{Global Forwarding (GloFo) Controller:}

\begin{theorem}
\label{thm:CLF_GloFo}
Consider the system \eqref{eq:unicycle_polar_closed_loop-Gv-1} in closed-loop with the feedback laws \eqref{eq-basic-v-control}, \eqref{eq-omega-general}, and
\begin{equation}
\label{eq:GloFo}
\tilde\omega = k_2\gamma + k_3\mbox{sinc}(2\gamma)
\left(\delta + \dfrac{k_1}{2k_2}\mbox{Si}(2\gamma) \right) \,,
\end{equation}
 with arbitrary $k_1, k_2, k_3 > 0$.
The origin $\rho = \delta = \gamma = 0$ is GAS on $\mathcal{S}$ in accordance with Definition \ref{def-our-GAS}.
Furthermore, the
continuously differentiable function 
\begin{equation}
V(\rho,\delta, \gamma)  = 
\rho^2+\left(\delta + \frac{k_1}{2k_2} {\rm Si}(2\gamma) \right)^2 + q^2 \gamma^2
\,,  \label{eq:V_CLF_GloFo_gamma_delta}
\end{equation}
with $q = \sqrt{k_1/k_3}$ is a (globally) strict CLF for \eqref{eq:unicycle_polar_closed_loop-Gv-1} with respect to the input pair $(v/\rho,\omega)$ in the sense of Def.~\ref{def-CLF}. 
\end{theorem}

\begin{proof}
Consider the model \eqref{eq:unicycle_polar_closed_loop-Gv-3-pass} and the 
forwarding transformation
\begin{equation}
\zeta =\delta+ \dfrac{k_1}{2k_2}\mbox{Si} (2\gamma) \,, \label{eq:forwarding_trans1}
\end{equation} with the function ${\rm Si}$ as defined in \eqref{eq:Si_function}.
Then, the transformed open-loop system is
\begin{eqnarray}
\dot\zeta &=& \frac{k_1}{k_2}\mbox{sinc}(2\gamma)(k_2\gamma -\tilde\omega) \,,
\\
\dot\gamma &=& - \tilde\omega\,.
\end{eqnarray}
Choosing the control law as in \eqref{eq:GloFo} yields
the closed-loop system
\begin{subequations}
    \label{eq:GloFo_temp}
\begin{eqnarray}
\dot\zeta &=& -\frac{k_1k_3}{k_2}\mbox{sinc}^2(2
\gamma) \zeta\,,
\\
\dot\gamma &=& - k_2\gamma - k_3\mbox{sinc}(2\gamma)\zeta\,. 
\end{eqnarray}
\end{subequations}
We choose the Lyapunov function in \eqref{eq:V_CLF_GloFo_gamma_delta}, which due to it being positive and radially unbounded for all $(\rho,\delta,\gamma) \ne (0,0,0)$ in $\{\rho \ge 0\} \times \mathcal{T}$,  satisfies $ \bar\alpha_1(|(\rho,\delta, \gamma)|_{\mathcal{S}}) \le V(\rho,\delta,\gamma) \le \bar\alpha_2(|(\rho,\delta, \gamma)|_{\mathcal{S}})$ for some class $\mathcal{K}_{\infty}$ functions $\bar{\alpha}_{1}$ and $\bar{\alpha}_{2}$ (see \cite[Lemma 4.3]{khalil_nonlinear_2002}). Hence, \eqref{eq:V_CLF_GloFo_gamma_delta} satisfies property~\ref{it:clf_prop_1} of Def.~\ref{def-CLF}. 
Its time derivative along the trajectories of \eqref{eq:unicycle_polar_closed_loop-Gv-2}, \eqref{eq:GloFo_temp} is 
\begin{equation}
\hspace*{-0.2cm}
\begin{aligned}[b]
\dot V   =- 2k_1 \rho^2 \cos^2 \gamma - &\frac{k_1k_2}{k_3} \left[
\left(\frac{k_3}{k_2}\mbox{sinc}(2\gamma) \zeta\right)^2 +\gamma^2 \right.  \\ &\left. + \left(\frac{k_3}{k_2}\mbox{sinc}(2\gamma) \zeta +\gamma \right)^2 
\right] \,, \label{eq:V_dot_Glofo}
\end{aligned}
\end{equation}
which is negative for all $(\rho, \delta, \gamma) \ne (0,0,0)$ in $\{\rho \ge 0\} \times \mathcal{T}$. With the control laws \eqref{eq-basic-v-control}, \eqref{eq-omega-general} and \eqref{eq:GloFo} as well as  \eqref{eq:V_dot_Glofo}, we show that property~\ref{it:clf_prop_2} of Def.~\ref{def-CLF} holds and thus \eqref{eq:V_CLF_GloFo_gamma_delta} is a CLF.
Then, again
based on \cite[Lemma 4.3]{khalil_nonlinear_2002},  
there exists 
$\alpha_3\in\mathcal{K}$, such that $\dot{V}(\rho,\delta,\gamma) \le - \alpha_3 \circ \bar \alpha_2^{-1}
(V)$, where $
\alpha_3 \circ \bar\alpha_2^{-1}\in\mathcal{K}$. 
Furthermore, this implies the existence of  $\beta\in\mathcal{KL}$such that
$V(t) \le \beta(V_0, t)$ 
 for all $t\geq 0$, according to \cite[Lemma 4.4]{khalil_nonlinear_2002}. The inequality $V(t) \le \beta(V_0, t)$ further implies that 
\begin{equation}
\abs{(\rho,\delta,\gamma)}_{\mathcal{S}} \le \bar{\alpha}_1^{-1}(\beta( \bar{\alpha}_2(\abs{(\rho_0, \delta_0,\gamma_0)}_{\mathcal{S}}),t)) \,, \label{eq:KL_estimate} 
\end{equation}
where $\bar{\alpha}_1^{-1}(\beta( \bar{\alpha}_2(r),t))$ is class $\mathcal{KL}$ in $(r,t)$.
Thus, 
$\rho =  \delta = \gamma = 0$ is GAS on $\mathcal{S}$.
\end{proof}

Theorem \ref{thm:CLF_GloFo} establishes global asymptotic stability in the polar coordinates $(\rho,\delta,\gamma)$. One should not expect that GAS also follows from Theorem~\ref{thm:CLF_GloFo} for the closed-loop system in the Cartesian coordinates $(x,y,\theta)$. However, attractivity of the point $(x,y,\theta)=(0,0,0)$ does follow. 

\begin{corollary}
\label{cor-attractivity}
\textit Consider the control system 
\eqref{eq:unicycle_cartesian} in closed loop with the feedback laws \eqref{eq-basic-v-control}, \eqref{eq-omega-general},  \eqref{eq:GloFo}.
For all initial conditions $(x_0,y_0,\theta_0)\in\mathbb{R}^3$ such that $x_0^2 + y_0^2>0$, 
\begin{align}
&\abs{x(t)} + \abs{y(t)} + \abs{\theta(t)}\nonumber\\
&\le  {\beta}(\abs{x_0} + \abs{y_0} + \abs{\theta_0}  + \abs{{\rm atan2}(y_0,x_0) + \pi},t)\,,  \label{eq:KL-estimate_xytheta}
\end{align}
for all $t\in[0,\infty)$, where $\beta\in\mathcal{KL}$  is defined as  $\beta(r,t) := \sqrt{2}\bar{\alpha}_1^{-1}(\beta( \bar{\alpha}_2(2 r),t))$.
\end{corollary}

This corollary follows from \eqref{eq:KL_estimate}, with the aid of the inequalities 
$\frac{1}{\sqrt{2}} (\abs{x} + \abs{y} + \abs{\theta}) \le \abs{(\rho,\delta,\gamma)}_{\mathcal{S}} \le \abs{x} + \abs{y} + \abs{\theta}  + 2\abs{\delta}$,  established with the inverse polar coordinate transformation (see Fig.~\ref{fig:unicycle_cord}).
The achievement of attractivity, without stability, in the Cartesian coordinates $(x,y,\theta)$, in Corollary \ref{cor-attractivity}, is not a shortcoming of our design method. It is consistent with the result of \cite[Remark 1.6]{coron1994relation}, proven also in \cite[pg. 43]{praly2022fonctions}, that the unicycle, in the Cartesian representation, is impossible to stabilize by static feedback, even if such feedback is permitted to be discontinuous, as are the feedback laws \eqref{eq-basic-v-control}, \eqref{eq-omega-general} and  \eqref{eq:GloFo} in the Cartesian coordinates $(x,y,\theta)$, as well as all the other feedback laws in this paper. 
\subsection{Bounded-in-LoS-Angle by Forwarding (BoFo) Controller:}

\begin{theorem}
\label{thm:CLF_BoFo}
Consider the system \eqref{eq:unicycle_polar_closed_loop-Gv-1} in closed-loop with the feedback laws \eqref{eq-basic-v-control}, \eqref{eq-omega-general}, and 
\begin{equation}
\label{eq:BoFo}
\tilde\omega = k_2\sin\gamma + k_3\dfrac{\cos\gamma}{\left(1 + \tan^2 \dfrac{\gamma}{2}\right)^2}
\left(\delta + \dfrac{k_1}{k_2}\sin\gamma \right)\,,
\end{equation}
with arbitrary $k_1, k_2, k_3 > 0$.
The origin $\rho = \delta = \gamma = 0$ is GAS on $\mathcal{S}_1$ in accordance with Definition \ref{def-our-GAS}.
Furthermore, the continuously differentiable function
\begin{equation}
V(\rho,\delta, \gamma)  = 
\rho^2 + \left(\delta + \frac{k_1}{k_2} \sin \gamma \right)^2+ 4q^2 \tan ^2\frac{\gamma}{2} 
\,,  \label{eq:CLF_BoFo_gamma_delta}
\end{equation}
with $q = \sqrt{k_1/k_3}$ is a (globally) strict CLF for \eqref{eq:unicycle_polar_closed_loop-Gv-1} with respect to the input pair $(v/\rho,\omega)$ in the sense of Def. \ref{def-CLF}. 
\end{theorem}

\begin{proof}
Consider the model \eqref{eq:unicycle_polar_closed_loop-Gv-3-pass}
with the forwarding transformation
\begin{equation}
\zeta = \delta + \frac{k_1}{k_2}\int_0^{\tan\frac{\gamma}{2}} \frac{\sin(4\arctan\sigma)}{2\sigma}{\rm d}\sigma = \delta + \dfrac{k_1}{k_2} \sin \gamma \,. \label{eq:forwarding_trans2}
\end{equation}
Then, the transformed open-loop system is
\begin{eqnarray}
\dot\zeta &=& \frac{k_1}{k_2}\cos\gamma\left(k_2\sin\gamma -
\tilde\omega \right) \,,
\\
\dot\gamma &=& -  \tilde\omega\,.
\end{eqnarray}
Considering the control law \eqref{eq:BoFo}, 
the closed-loop system is
\begin{subequations}
    \label{eq:BoFo_temp}
\begin{eqnarray}
\dot\zeta &=& -\frac{k_1k_3}{k_2}\dfrac{\cos^2(\gamma)}{\left(1+\tan^2\frac{\gamma}{2}\right)^2} \zeta \,,
\\
\dot\gamma &=& - k_2\sin\gamma - k_3\dfrac{\cos(\gamma)}{\left(1+\tan^2\frac{\gamma}{2}\right)^2}\zeta\,.
\end{eqnarray}
\end{subequations}
We choose the Lyapunov function \eqref{eq:CLF_BoFo_gamma_delta} which is positive and radially unbounded for all $(\rho,\delta,\gamma) \ne (0,0,0)$ in $\{\rho \ge 0\} \times \mathcal{T}_1$, and satisfies $ \bar\alpha_1(|(\rho,\delta, \gamma)|_{\mathcal{S}_1}) \le V(\rho,\delta,\gamma) \le \bar\alpha_2(|(\rho,\delta, \gamma)|_{\mathcal{S}_1})$ for some class $\mathcal{K}_{\infty}$ functions $\bar{\alpha}_{1}$ and $\bar{\alpha}_{2}$.
Its time derivative along the trajectories of \eqref{eq:unicycle_polar_closed_loop-Gv-1}, \eqref{eq:BoFo_temp} is
\begin{equation}
\begin{aligned}[b]
\dot V  = &-2k_1 \rho^2 \cos^2 \gamma- \frac{k_1k_2}{k_3}\Biggl[
\left(\frac{k_3}{k_2}\frac{\cos\gamma}{{1+\tan^2\frac{\gamma}{2}}} \zeta\right)^2\\
&+4\tan^2\frac{\gamma}{2}
+ \left(\frac{k_3}{k_2}\frac{\cos\gamma}{{1+\tan^2\frac{\gamma}{2}}}\zeta +2\tan\frac{\gamma}{2} \right)^2
\Biggr]\label{eq:V_dot_forward2}
\end{aligned}\,,
\end{equation}
which is negative for all $(\rho,\delta,\gamma) \ne (0,0,0)$ in $\{\rho \ge 0\} \times \mathcal{T}_1$.
Using similar arguments as in the proof of Theorem~\ref{thm:CLF_GloFo}, we can show that \eqref{eq:CLF_BoFo_gamma_delta} is a CLF in the sense of Def.~\ref{def-CLF}, as well as 
$\abs{(\rho,\delta,\gamma)}_{\mathcal{S}_1} \le \bar{\alpha}_1^{-1}(\beta( \bar{\alpha}_2(\abs{(\rho_0, \delta_0,\gamma_0)}_{\mathcal{S}_1}),t))$, 
where $\bar{\alpha}_1^{-1}(\beta( \bar{\alpha}_2(r),t))$ is class $\mathcal{KL}$ in $(r,t)$.
Consequently, 
$\rho =  \delta = \gamma = 0$ is GAS on $\mathcal{S}_1$.
\end{proof}

\subsection{On Global Forwarding Designs with Actuated Velocity}

\begin{figure}[t]
\centering
\begin{subfigure}[b]{\linewidth}
\centering
\includegraphics[width=.9\linewidth]{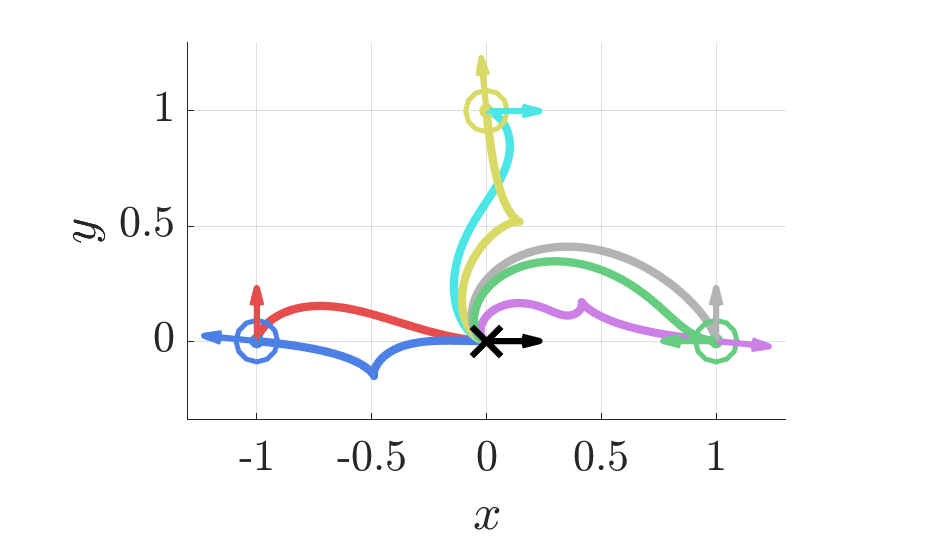}
\caption{Robot Cartesian trajectories for multiple representative initial conditions. Notably, the blue, yellow, and purple initial conditions are slightly offset from $|\gamma(0)| = \pi$ due to the bounded LoS angle constraint.}
\vspace{0.5em}
\label{fig:trajectory_bofo}
\end{subfigure}
\begin{subfigure}[b]{\linewidth}
\centering
\includegraphics[width=0.9\linewidth]{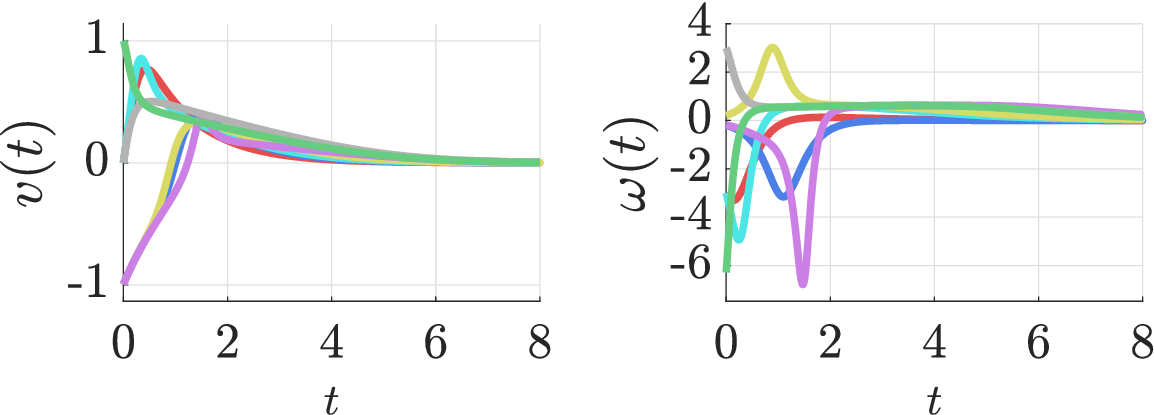}
\caption{Forward velocity $v$ and steering input $\omega$.}
\vspace{0.5em}
\label{fig:control_bofo}
\end{subfigure}
\caption{Simulation with the BoFo control law~\eqref{eq:BoFo} along with \eqref{eq-basic-v-control} and \eqref{eq-omega-general} from Theorem~\ref{thm:CLF_BoFo} with $(k_1,k_2,k_3,k_4) = (1,3,2,1.5)$.}
\label{fig:sim_bofo}
\end{figure}

The globally asymptotically stabilizing forwarding controller \eqref{eq:GloFo}, referred to as GloFo, has an analogous backstepping version, the GloBo controller \cite{todorovski2025_CLF}, the region of attraction of which is \eqref{eq:ss_S}. The simulation results are almost identical to the GloBo control law in \cite{todorovski2025_CLF} and hence omitted.

The backstepping method developed in \cite{todorovski2025_CLF} does not allow for the construction of barrier CLFs that ensure invariance of sublevel sets in which the LoS angle $\gamma$ remains bounded. To address this limitation, we introduce the forwarding-based controller \eqref{eq:BoFo} which, in analogy to the passivity-based BoLSA controller in \cite{todorovski2025_CLF}, explicitly bounds the LoS angle $\gamma$. For this reason, it is referred to as the Bounded LoS angle by Forwarding (BoFo) controller. Much like BoLSA, BoFo has the region of attraction \eqref{eq:ss_S1}, ensuring that the unicycle never fully turns its back to the goal position while parking, as long as the initial heading is not exactly opposite to the target. The numerical simulations for the BoFo controller are given in Fig.~\ref{fig:sim_bofo}. Interestingly, with the forwarding technique, the polar angle $\delta$ cannot be bounded, because it enters linearly in the forwarding transformations \eqref{eq:forwarding_trans1} and \eqref{eq:forwarding_trans2}, making the limitations of forwarding and backstepping mirror each other. 

The modular design framework of Section~\ref{sec:modular_section} plays a crucial role: it reveals the strict feedforward structure of \eqref{eq:unicycle_polar_closed_loop-Gv-3-pass}, which makes integrator forwarding applicable. Consequently, the Lyapunov functions \eqref{eq:V_CLF_GloFo_gamma_delta} and \eqref{eq:CLF_BoFo_gamma_delta}—that is, the CLFs associated with GloFo and BoFo without the $\rho^2$ term—can be incorporated into a composite Lyapunov framework \cite{todorovski2025_CLF}. This not only enriches the pool of available CLFs but also facilitates the design of inverse optimal controllers, as exploited in \cite{Kim2025_IOC}.

\section{Deadbeat Parking with Integrator Forwarding}\label{sec-forward}

Fixing the forward velocity $v$ to a positive constant gives \eqref{eq:unicycle_polar_closed_loop-Gv-1} a strict feedforward form, which allows the application of integrator forwarding to the whole system. Since \eqref{eq:unicycle_polar_closed_loop-Gv-3} is independent of $\rho$, the steering can be designed modularly, as in the previous sections. This enables the development of finite-time stabilizing control laws—achieving deadbeat parking—for the Dubins vehicle. While several such laws have been introduced in \cite{Krstic2025_Dubins}, we present additional results by employing integrator forwarding. In this regard, we first recall the following lemmas.

\begin{lemma}[\!\cite{Krstic2025_Dubins}]
\label{lem1}
Let $a>0$ and let $V:[0,\rho_0]\to\mathbb{R}_{\ge 0}$ be a continuously differentiable function. Then, for $\rho\in(0,\rho_0]$,
\begin{itemize}
\item $\dfrac{dV}{d\rho} \geq \dfrac{a}{\rho}V$ implies $V(\rho)\leq V(\rho_0) \left(\dfrac{\rho}{\rho_0}\right)^a$\,;
\vspace{0.2cm}
\item $\dfrac{dV}{d\rho} \geq \dfrac{a}{\rho^2}V$ implies $V(\rho)\leq V(\rho_0) {\rm e}^{a\left(1/{\rho_0}-{1}/{\rho}\right)}$.
\end{itemize}
\end{lemma}
\vspace{0.2cm}
\begin{lemma}[\!\cite{Krstic2025_Dubins}]
\label{lem2}
Consider $\dot\rho = - v\cos\gamma$, for $v>0$. Let $\alpha\in\mathcal{K}$ and let
\begin{equation}
t_1 = \frac{\rho_0}{v}\sqrt{1+\alpha(1)}\,.
\end{equation}
For all $t\in [0,T)$  for which the solution exists, if $\cos\gamma(t)>0$ and $\tan^2\gamma(t)\leq \alpha(\rho(t)/\rho_0)$ for all $t\in[0,T)$, then $\rho(t)\leq \rho_0(1-t/t_1)$ for all $t\in[0,T)$.
\end{lemma}


\begin{theorem}
\label{thm:Dubins-FT-stabilize6}
Let $v=\mbox{\rm const}>0$. Consider \eqref{eq:unicycle_polar_closed_loop-Gv-1}
with the control law 
\begin{equation}
    \omega =\frac{v}{\rho} \left(\sin\gamma +\cos^3\gamma \, \bar\omega\right)
\end{equation}
and
\begin{subequations}
\label{eq-third-controller6}
\begin{eqnarray}
\label{eq-third-controller26}
\bar\omega &=& c_1\tan\gamma+c_2\zeta \,,
\\
\label{eq-third-controller36}
\zeta &=& \tan\gamma + c_1\delta\,,
\end{eqnarray}
\end{subequations}
with $\underline c :=\min\{c_1,c_2\} > 2$. 
For all $\rho_0>0$, $\delta_0\in\mathbb{R}$, and $\gamma_0\in(-\pi/2,\pi/2)$ the following holds:
\begin{equation}
\label{eq2-rho-bound6}
\rho(t)\leq \rho_0 (1-t/t_1) \,,
\end{equation}
\begin{equation}
\label{eq2-deltan-bound6}
B^2(t) \leq 
2c_1 c_2 ({1-t/t_1})^{\underline c}
B_0^2 \,,
\end{equation}
\begin{align}
|\omega(t)| \leq \frac{v}{\rho_0}&(1+c_1 +c_2+c_1c_2)\times \nonumber \\ &\sqrt{2c_1 c_2}  (1-t/t_1)^{\underline c/2 -1} \label{eq2-omega(t)-bound6}
B_0 \,,
\end{align}
for all $t\in\left[0, \min\left\{t_1,T\right\} \right)$, where 
\begin{equation}
\label{eq-t1N1beta16}
t_1 =\frac{\rho_0}{v}{\sqrt{1+2c_1 c_2B_0^2}}\,,
\end{equation}
and $T$ is the interval of existence of the solutions, that is $\rho(T)=0$, and $B_0 = B(\delta_0,\gamma_0)$, with $B(\delta,\gamma) \coloneqq \sqrt{\delta^2 + \tan^2\gamma}$.
\end{theorem}

\begin{figure}[t]
\centering
\begin{subfigure}[b]{\linewidth}
\centering
\includegraphics[width=.85\linewidth]{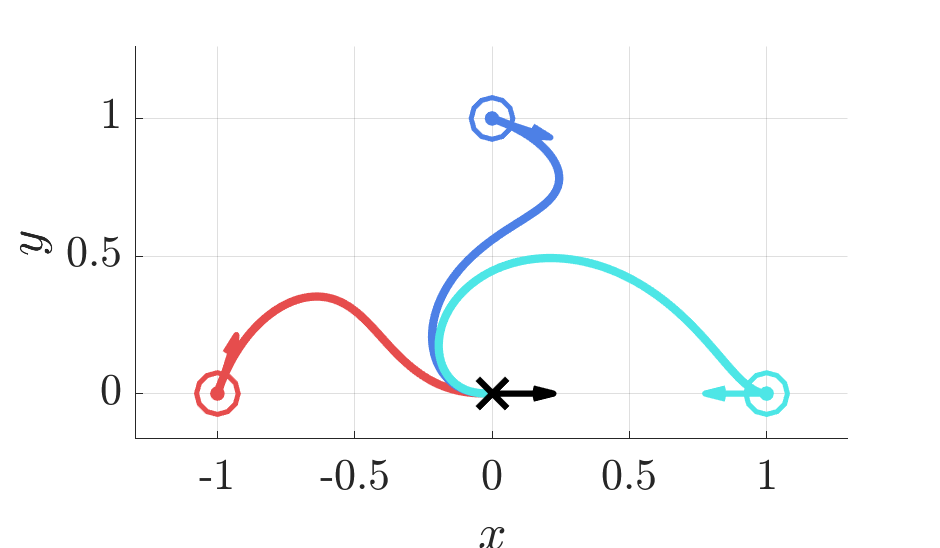}
\caption{Robot Cartesian trajectories for the initial conditions $[\rho(0), \delta(0), \gamma(0)] = [1, 0,-\pi/2.5]$ (red), $[1,-\pi/2,-\pi/2.5]$ (blue) and $[1, \pi, 0]$ (cyan).}
\vspace{0.5em}
\label{fig:trajectory_thrm1}
\end{subfigure}
\begin{subfigure}[b]{\linewidth}
\centering
\includegraphics[width=0.85\linewidth]{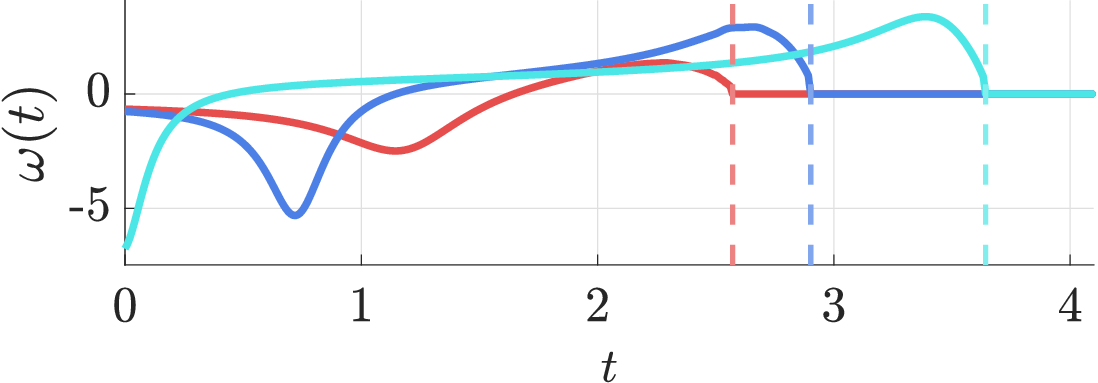}
\caption{Steering input $\omega$ with cutoff ($v = \omega = 0$) applied when $\rho \leq 0.01$. The time at which the cutoff condition is met is indicated by the dashed vertical line.}
\vspace{0.5em}
\label{fig:control_thrm1}
\end{subfigure}
\vspace{0.5cm}
\begin{subfigure}[b]{\linewidth}
\centering
\includegraphics[width=.9\linewidth]{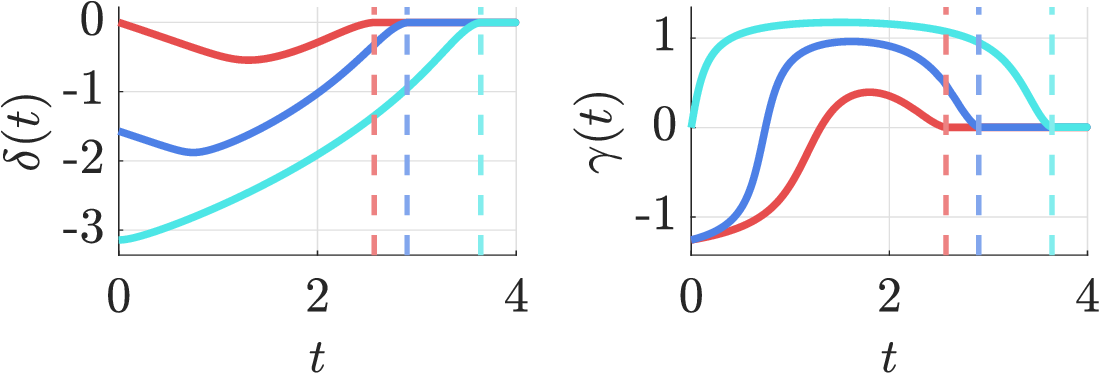}
\caption{Plots of the polar angle $\delta$ and of the LoS angle $\gamma$ versus time, indicating that angular states reach zero before the cutoff.}
\label{fig:polar_thrm1}
\end{subfigure}
\caption{Simulations with the steering control in Theorem~\ref{thm:Dubins-FT-stabilize6}, with $c_1 = 2.05$, $c_2 = 2.1$, and $v = 0.5$. Since the forward velocity is constrained to a nonzero constant, the controllers must be turned off once $\rho(T) = 0$. To mitigate numerical inaccuracies near the boundary of the interval of existence, we impose a cutoff on both control inputs if $\rho(t) \leq 0.01$.}
\label{fig:sim_thrm1}
\end{figure}

\begin{proof} 
We outline the proof by first noting that the forwarding transformation and controller yield
\begin{eqnarray}
 \dot{\zeta} &=& -\frac{v\cos\gamma}{\rho} c_2\zeta \,,
\\
\dot{\wideparen{\tan\gamma}} &=& -\frac{v\cos\gamma}{\rho}(c_1\tan\gamma+c_2\zeta)\,.
\end{eqnarray}
Note now that, the Lyapunov function $V=\frac{c_2}{c_1}\zeta^2+\tan^2\gamma$ satisfies
\begin{equation}
\begin{aligned}[b]
    \dot V = - \frac{v\cos\gamma}{c_1\rho}&\left[
(c_2\zeta)^2 +(c_1\tan\gamma)^2 \right. \\ &\left. + (c_2\zeta+c_1\tan\gamma)^2
\right]\,.
\end{aligned}
\end{equation}
We then obtain $dV/d\rho \geq \underline c V/\rho$ and hence, by Lemma \ref{lem1},  $V(\rho) \leq (\rho/\rho_0)^{\underline c} V(\rho_0)$, for $\rho\in(0,\rho_0]$. Then, from the definitions of the forwarding transformation and of the Lyapunov function, 
one can show that $B^2(\delta,\gamma) \leq \lambda_1 V(\rho)$ and $V(\rho)\leq \lambda_2 B_0^2$. It then follows that $B^2(\delta,\gamma)\leq 2c_1 (\rho/\rho_0)^{\underline c} B_0^2$. Since $\tan^2\gamma\leq 2c_1 (\rho/\rho_0)^{\underline c}  B_0^2$, by Lemma~\ref{lem2}, the inequalities \eqref{eq2-rho-bound6} and \eqref{eq2-deltan-bound6} follow.  From \eqref{eq-third-controller6} we get that $|\omega| \leq \frac{v}{\rho}(1+c_1+ c_2+c_1 c_2) B\leq \frac{v}{\rho_0}(1+c_1+ c_2+c_1 c_2)\sqrt{2c_1c_2}(\rho/\rho_0)^{\underline c/2 -1} B_0$, from which \eqref{eq2-omega(t)-bound6} follows. 
\end{proof}

It is important and interesting to note that the forwarding controller in Theorem \ref{thm:Dubins-FT-stabilize6} is virtually identical to the controller in \cite[Thm.~1]{Krstic2025_Dubins}, although their design processes are vastly different. Both controllers employ the state transformation $\zeta=\tan\gamma + c_1\delta$, which is in one case a backstepping transformation and in the other case the forwarding transformation. Nevertheless, the controllers differ minimally:
\begin{equation}
\bar\omega_{\rm bkst} = \bar\omega_{\rm fwd} +\delta \,, \quad \bar\omega_{\rm fwd} = c_1\tan\gamma+c_2\zeta\,.
\end{equation}
The difference of $\delta$ accounts for the difference in the target system, which for backstepping induces skew-symmetry of the target states $(\delta,\zeta)$ and for forwarding a cascade structure of the target states $(\zeta,\tan\gamma)$. This resemblance is shown in Fig.~\ref{fig:sim_thrm1}, where the system behavior closely matches that resulting from the control law in~\cite[Thm.~1]{Krstic2025_Dubins}. Moreover, the nonsmooth nature of the controller is also similar, as illustrated in Fig.~\ref{fig:control_thrm1}, and motivates the development of the forwarding version of \cite[Thm.~2]{Krstic2025_Dubins} which shows the smooth-in-time exponential convergence of $\delta$, $\tan\gamma$ and  $\omega$ as $t$ approaches $t_1$.

\begin{theorem}
\label{thm:Dubins-FT-stabilize7}
Let $v=\mbox{\rm const}>0$. Consider \eqref{eq:unicycle_polar_closed_loop-Gv-1}
with the control law  
\begin{equation}
    \omega =\frac{v}{\rho} \left(\sin\gamma +\cos^3\gamma \, \bar\omega\right)\,,
\end{equation}
where
\begin{subequations}
\label{eq-third-controller67}
\begin{eqnarray}
\label{eq-third-controller267}
\bar\omega &=& \frac{1}{\rho}[c_1(\tan\gamma +\delta)+c_2\zeta] +\tan\gamma\,,
\\
\label{eq-third-controller37}
\zeta &=& \tan\gamma +\delta +\frac{c_1}{\rho}\delta\,,
\end{eqnarray}
\end{subequations}
and $c_1,c_2 \geq \underline c :=\min\{c_1,c_2\} > 0$. Then, there exist 
constants $\beta_1\geq \beta_2 >0$ and $N_2(\rho_0)\geq N_1(\rho_0) >0$, which also depend on $c_1,c_2$, 
such that for all $\rho_0>0$, $\delta_0\in\mathbb{R}$, and $\gamma_0\in(-\pi/2,\pi/2)$ the following holds:
\begin{equation}
\label{eq2-rho-bound}
\rho(t)\leq \rho_0 (1-t/t_1) \,,
\end{equation}
\begin{equation}
\label{eq2-deltan-bound}
B^2(t) \leq 
N_1 {\rm e}^{-\beta_1/({1-t/t_1})} 
B^2_0 \,,
\end{equation}
\begin{align}
\label{eq2-omega(t)-bound}
|\omega(t)| \leq v \, 
N_2 {\rm e}^{-\beta_2/({1-t/t_1})}
B_0 \,,
\end{align}
for all $t\in\left[0, \min\left\{t_1,T\right\} \right)$, where
\begin{equation}
\label{eq-t1N1beta1}
t_1(\rho_0,\delta_0, \gamma_0, v,c_1,c_2) =\frac{\rho_0}{v}{\sqrt{1+N_1 {\rm e}^{-\beta_1}B^2_0}}\,,
\end{equation}
and $T$ is the interval of existence of the solutions, that is $\rho(T)=0$, and $B_0 = B(\delta_0,\gamma_0)$, with $B(\delta,\gamma) \coloneqq \sqrt{\delta^2 + \tan^2\gamma}$.
\end{theorem}

\begin{proof}
The outline of the proof is as follows. With $\Gamma=\tan\gamma +\delta$ and the controller, one can verify that
\begin{eqnarray}
\frac{{\rm d}\zeta}{{\rm d}\rho }&=& \frac{1}{\rho^2}c_2 \zeta \,,
\\
\frac{{\rm d}\Gamma}{{\rm d}\rho} &=& \frac{1}{\rho^2}(c_1\Gamma + c_2\zeta)\,.
\end{eqnarray}
The Lyapunov function $V=\frac{c_2}{c_1}\zeta^2+\Gamma^2$ satisfies 
\begin{equation}
\hspace*{-0.3cm}
\frac{dV}{d\rho} = \frac{1}{c_1\rho^2}\left[
(c_2\zeta)^2 +(c_1\Gamma)^2 + (c_2\zeta+c_1\Gamma)^2
\right]\geq \underline c \frac{ V}{\rho^2}\,,
\end{equation} 
with $\underline c :=\min\{c_1,c_2\}$, and hence, by Lemma \ref{lem1}, $V(\rho) \leq {V(\rho_0)}{{\rm e}^{\underline c \left(\frac{1}{\rho_0}-\frac{1}{\rho}\right)}}$ for $\rho\in(0,\rho_0]$. 
Then, from the definitions of the forwarding transformation and of the Lyapunov function,
one can show that there exists $\beta_1 > 0$ such that ${\rm e}^{\underline c \left(\frac{1}{\rho_0}-\frac{1}{\rho}\right)} \le {\rm e}^{- \beta_1 \rho_0 / \rho}$, as well as $N_1(\rho)$ such that 
\begin{equation}
    B^2(\delta,\gamma) \leq N_1(\rho_0) {\rm e}^{- \beta_1 \rho_0 / \rho} B_0 \,,\label{eq:temp_B_bound}
\end{equation}
 for all $\rho \in (0, \rho_0]$.
With Lemma~\ref{lem2} one obtains \eqref{eq2-rho-bound}. Substituting \eqref{eq2-rho-bound} in \eqref{eq:temp_B_bound} and noting from \eqref{eq2-rho-bound} that $-\rho_0/\rho(t) \le -1 / (1-t/t_1)$ one gets \eqref{eq2-deltan-bound}.
Similarly, one can show that there exists $N_2(\rho) \ge N_1(\rho) > 0$, $\beta_2 > 0$ such that \eqref{eq2-omega(t)-bound} holds.
\end{proof}

Figure~\ref{fig:sim_thrm2} presents the numerical simulation corresponding to Theorem~\ref{thm:Dubins-FT-stabilize7}. The trajectories converge to the origin and, much like the behavior resulting from the application of the control law in~\cite[Thm.~2]{Krstic2025_Dubins}, the steering smoothly decays to zero before the cutoff condition is reached.

\begin{figure}[t]
\centering
\begin{subfigure}{\linewidth}
\centering
\includegraphics[width=0.8\linewidth]{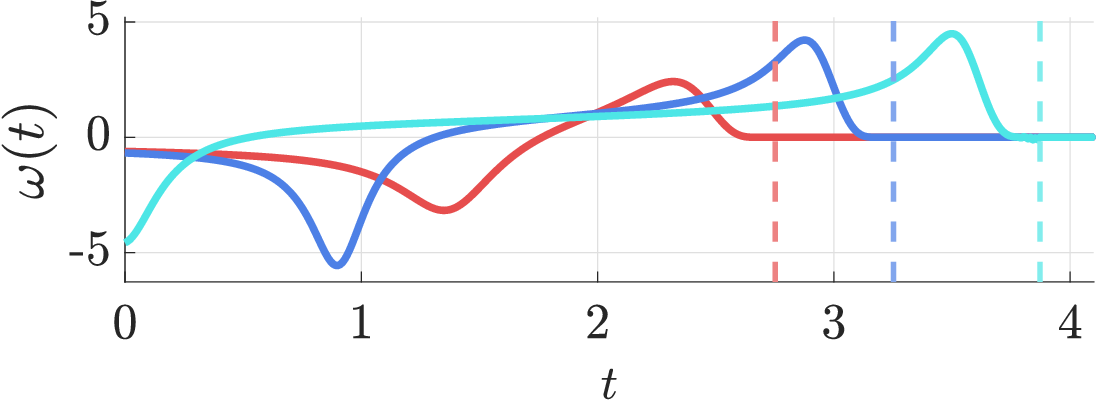}
\caption{Steering input $\omega$ with cutoff ($v = \omega = 0$) applied when $\rho \leq 0.01$. The time at which the cutoff condition is met is indicated by the dashed vertical line.}
\vspace{0.5em}
\label{fig:control_thrm2}
\end{subfigure}
\begin{subfigure}{\linewidth}
\centering
\includegraphics[width=.9\linewidth]{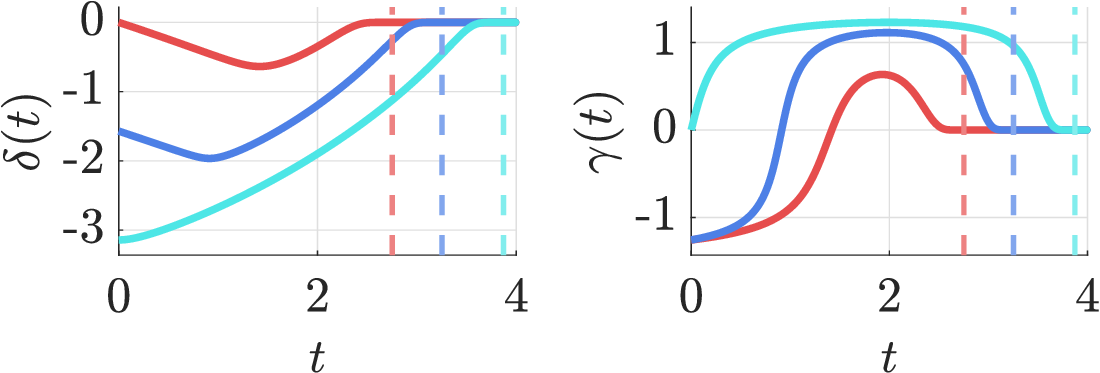}
\caption{Plots of the polar angle $\delta$ and of the LoS angle $\gamma$, indicating that angular states reach zero before the cutoff.}
\label{fig:polar_thrm2}
\end{subfigure}
\caption{Simulation with the steering control law in Theorem~\ref{thm:Dubins-FT-stabilize7} with $c_1 = 0.7$,  $c_2 = 1.3$ and $v = 0.5$. The trajectory in the $xy$-plane is nearly identical to that shown in Fig.~\ref{fig:trajectory_thrm1}, differing only in minor details, and is therefore omitted.}
\label{fig:sim_thrm2}
\end{figure}

\bibliographystyle{IEEEtranS}
\bibliography{root}

\clearpage

\end{document}